\documentclass[sigconf]{aamas} 

\usepackage{balance} 

\setcopyright{ifaamas}
\acmConference[AAMAS '23]{Proc.\@ of the 22nd International Conference
on Autonomous Agents and Multiagent Systems (AAMAS 2023)}{May 29 -- June 2, 2023}
{London, United Kingdom}{A.~Ricci, W.~Yeoh, N.~Agmon, B.~An (eds.)}
\copyrightyear{2023}
\acmYear{2023}
\acmDOI{}
\acmPrice{}
\acmISBN{}

\usepackage[english]{babel}
\usepackage{float}
\usepackage{multirow}
\usepackage{amsmath}
\usepackage{graphicx}
\usepackage{amsmath,amsfonts, bm}

\DeclareMathOperator{\BR}{\mathbb R}
\DeclareMathOperator{\col}{\mathrm{col}}
\DeclareMathOperator{\diag}{\mathrm{diag}}

\newtheorem{lemma}{\bf Lemma}
\newtheorem{assumption}{\bf Assumption}
\newtheorem{theorem}{\bf Theorem}
\newtheorem{remark}{\bf Remark}

\usepackage[absolute,overlay]{textpos}
\newcount\Comments  
\usepackage{color}
\definecolor{darkgreen}{rgb}{0,0.5,0}
\definecolor{purple}{rgb}{1,0,1}
\newcommand{\kibitz}[2]{\ifnum\Comments=1\textcolor{#1}{#2}\fi}

\acmSubmissionID{???}

\title[AAMAS-2023 Formatting Instructions]{Decentralised and Cooperative Control of Multi-Robot Systems through Distributed Optimisation}

\author{Yi Dong$^\S$, Zhongguo Li$^\dag$, Xingyu Zhao$^\S$, Zhengtao Ding$^\ddag$, Xiaowei Huang$^\S$}
\affiliation{
\institution{$^\S$University of Liverpool\country{United Kingdom}, \{yi.dong, xingyu.zhao,  xiaowei.huang\}@liverpool.ac.uk}}

\affiliation{
  \institution{$^\dag$University College London \country{United Kingdom}, zhongguo.li@ucl.ac.uk}}

\affiliation{
  \institution{$^\ddag$University of Manchester\country{United Kingdom}, zhengtao.ding@manchester.ac.uk}}

\begin{abstract}
Multi-robot cooperative control has gained extensive research interest due to its wide applications in civil, security, and military domains. This paper proposes a cooperative control algorithm for multi-robot systems with general linear dynamics. The algorithm is based on distributed cooperative optimisation and output regulation, and it achieves global optimum by utilising only information shared among neighbouring robots. Technically,  a high-level distributed optimisation algorithm for multi-robot systems is presented, which will serve as an optimal reference generator for each individual agent. Then, based on the distributed optimisation algorithm, an output regulation method is utilised to solve the optimal coordination problem for general linear dynamic systems. The convergence of the proposed algorithm is theoretically proved. Both numerical simulations and real-time physical robot experiments are conducted to validate the effectiveness of the proposed cooperative control algorithms.
\end{abstract}

\keywords{Cooperative Control, Distributed Optimisation, Optimal Coordination, Multi-robot Systems, Output Regulation}

\newcommand{\BibTeX}{\rm B\kern-.05em{\sc i\kern-.025em b}\kern-.08em\TeX}

\begin{document}


\pagestyle{fancy}
\fancyhead{}

\begin{textblock*}{20cm}(1cm,1cm)
	\textcolor{red}{{\large Accepted by The 22nd International Conference on Autonomous Agents and Multiagent Systems (AAMAS 2023).}}
\end{textblock*}

\maketitle 


\section{Introduction}

Cooperation of multiple robots to accomplish complex and challenging tasks \cite{ren2022event,yang2019survey} has been made possible by recent advances in high-performance computing, fast communication, and affordable onboard sensors.   
Nevertheless, it remains a challenge  to design decentralised algorithms for real-world multi-robot systems. {In this paper, we design an output-regulation based distributed optimisation algorithm to control the physical multi-robot systems.}



On distributed multi-agent systems, several research topics, including target aggregation, trajectory tracking, containment and formation control, can be formulated 
as consensus problems \cite{qin2016recent}. However, when it comes to the solutions to the distributed consensus, existing methods are either designed for overly-simplified problems or based on overly-simplified models of the agents. For example, in many settings, the consensus value is defined with respect to the initial states of the agents \cite{shi2009global,zou2021sampled,martinovic2022cooperative,wang2018prescribed,kuriki2014consensus}, and even if this constraint was relaxed, the agents may be modelled as a single-integrator (see Equation~\eqref{equ:single-integrator} for definition) \cite{li2021distributed, zuo2019fixed, cherukuri2016initialization, yi2016initialization}. While imposing stronger constraints 
leads to simpler mathematical proof on the theoretical results (such as convergence and correctness), the unrealistic constraints may result in a significant gap between theoretical guarantees and practical utility. 

In this paper, to address the cooperation between multiple robots, we relax the constraints from two perspectives. First, we work with consensus problems whose consensus value is defined over reference points, instead of initial states. A reference point may or may not be an initial state. Such generalisation is of practical importance, because in many applications the agents are initialised randomly and the agents' goals might be independent from their initial states. {Second, we consider linear systems, which generalise the single-integrator model by defining the dynamics of agents with linear operators (see Equations~\eqref{eqn: 2b} and \eqref{eqn: 2c} for definition). }

While considering a more general setting, we show that the theoretical guarantees are not compromised. That is, our algorithm can achieve both correctness (Section~\ref{parta}) and convergence (Sections \ref{partb} and \ref{partc}). This is owing to the novel distributed optimisation algorithms proposed in the paper. 
%
Distributed optimisation aims to solve an optimisation problem where the global cost function is composed of a set of $N$ local objectives $f_i(y)$ \cite{Li2021Automatica}, i.e., $\min_{y} \sum_{i=1}^{N} f_i (y) $. Due to limited communication and the requirement of local privacy protection, the local objective function $f_i(y)$ is only known to agent $i$. In distributed optimisation, each agent can only cooperate with its neighbours by exchanging non-sensitive information. 
Although this is not a new problem, existing methods \cite{qiu2016distributed,li2022exponential,tran2019distributed} either assume single-integrator models for agents or are based on continuous systems. {For robotic systems, particularly the high-level control of the robotic systems as we aim to address in this paper, the discrete-time system and heterogeneous assumption on the robotic systems are arguably more realistic, while algorithms and their associated proofs cannot be easily transferred
, which motivates our work.} Intuitively, our algorithm proceeds by every agent moving according to a recursive expression (i.e., considering not only on the current time but also the history) about a gradient over its local objective and the information collected from its neighbours (see Equation~\eqref{eqn: controller linear system}). {Also, we utilise distributed output regulation techniques to ensure that the formulated linear multi-robot systems can track the dynamic references in real-time \cite{ding2003global,ding2013back}.}

The proposed algorithm is implemented for a consensus control problem using a physical multi-robot system consisting of 4 Turtlebot robots. {The communication graph of the physical robots is designed as an undirected ring. Each robot is controlled and optimised only based on the information from itself and connected neighbours.
Although different robots have their respective local targets, they are eventually 
moved to the global optimal point since the final consensus value is generated by solving the real-time optimisation problems, which is independent of the initial states.}


The major contributions of 
this work
are summarised as follows:
\begin{enumerate}
    \item A distributed discrete-time cooperative control algorithm is proposed and the convergence of our algorithm has been theoretically proved. Different from most existing studies, e.g., \cite{ren2007information, qu2009cooperative} and references therein, which are extensively concentrated on continuous-time systems, the proposed algorithm in this paper is ready for implementation on digital robots and platforms. 
    \item {A composite approach combining distributed optimisation and output regulation is developed for heterogeneous linear systems. Different from the initialisation-dependent consensus problems, the proposed approach lays the foundation for the interaction between optimisation and control. } 
    \item {The proposed algorithm has been successfully validated on a real multi-robot system, where the errors and noises are handled by the communication among different agents. Furthermore, the reproducibility and replicability of our work are guaranteed since all source codes are available on Github for open access.}
\end{enumerate}

The rest of this paper is organised as follows. The mathematical preliminaries are summarised and the researched problem is also formulated in Section \ref{sec_pre}. An output regulation based distributed optimisation approach for multi-robot consensus protocol is proposed in Section \ref{sec_algorithm}. Simulation results and corresponding analysis are presented in Section \ref{sec_sim}. Finally, Section \ref{sec_con} concludes this paper.

\section{Related Work}
{This paper studies the control algorithm of multi-agent systems. The distributed algorithm designed in this paper is based on the consensus problem that requires a distributed protocol to drive a group of agents to achieve an agreement on states.}

{Initiated from \cite{olfati2004consensus}, consensus based distributed cooperative control problems have been widely studied and tactfully generalised to different specific sub-problems in recent years, such as finite-time, event-triggered, time-delayed, and switching-topology based cooperation problems. 
Shi and Hong considered the coordination problem of aggregation to a convex target set for a multi-agent system \cite{shi2009global}. 
Zou \textit{et al.} studied the coordinated aggregation problem of a multi-agent system while considering communication delays and applying a projection operator to guarantee the final consensus value within a target area \cite{zou2021sampled}.
Martinovi\'{c} \textit{et al.} proposed a distributed observer-based control strategy to solve the leader-following tracking problem \cite{martinovic2022cooperative}. 
Wang \textit{et al.} presented a distributed consensus and containment algorithm for the finite-time control of a multi-agent system based on time-varying feedback gain \cite{wang2018prescribed}.
Kuriki and Namerikawa illustrated a consensus-based cooperative formation control strategy with collision-avoidance capability for a group of multiple unmanned aerial vehicles \cite{kuriki2014consensus}. 
For all the above-mentioned solutions, the consensus value is hinged on the initial states of the agents, for example, the midpoint of the agents' initial locations. This is a severe restriction as in many practical scenarios, agents are initialised randomly and the goal of cooperation is without any correlation with the initial states. }

{To eliminate the initialisation step, the multi-agent consensus problem becomes a distributed optimisation problem when the consensus value is required to minimise the sum of local cost functions known to the individual agents. 
Qiu \textit{et al.} proposed a distributed optimisation protocol to minimise the aggregate cost functions while considering both constraint and optimisation \cite{qiu2016distributed}.
Li \textit{et al.} designed a proportional–integral (PI) controller to solve the optimal consensus problem and introduced event-triggered communication mechanisms to reduce the communication overhead \cite{li2022exponential}. 
Tran \textit{et al.} investigated two time-triggered and event-triggered distributed optimisation algorithms to reduce communication costs and energy consumption \cite{tran2019distributed}. 
Ning \textit{et al.} studied a fixed-time distributed optimisation protocol to guarantee the convergence within a certain steps for multi-agent system \cite{ning2017distributed}.}


In practice, it is desirable to realise consensus in discrete-time domain for real robotic applications. However, the aforementioned distributed control and optimisation works \cite{shi2009global,zou2021sampled,martinovic2022cooperative,wang2018prescribed,qiu2016distributed,li2022exponential,tran2019distributed,zuo2018adaptive,dong2019demand,dong2022short,wang2015distributed,zhao2017distributed,nedic2009distributed,nedic2010constrained,nedic2014distributed,ning2020bipartite,ning2017distributed,li2011consensus,wang2010distributed} can only works for either continuous or non-heterogeneous systems. 
As a result, it is difficult to guarantee the convergence of heterogeneous linear systems without initialisation in practical scenarios. 
Motivated by the observations above, in this paper, 
{the aim of this work is to solve distributed optimal coordination problem for discrete-time and heterogeneous linear systems. We rigorously prove that consensus of heterogeneous linear systems can be achieved, while the global costs are minimised. By the interdisciplinary nature of the target problem, it spans across different problem specifics including single integrate/linear system, continuous-time/discrete-time, and heterogeneous/non-heterogeneous. We 
summarise and list a few related references papers to Table. \ref{table:scope} to highlight the scope and position of this paper.}
\begin{table}[]
\caption{{Scope of the proposed method.}}
\resizebox{\columnwidth}{!}{
\label{table:scope}
\begin{tabular}{|l|ll|ll|}
\hline
\multirow{2}{*}{} & \multicolumn{2}{l|}{Single Integrator}                                                           & \multicolumn{2}{l|}{Linear System}                                           \\ \cline{2-5} 
                  & \multicolumn{1}{l|}{Continuous-time} & Discrete-time                                                       & \multicolumn{1}{l|}{Continuous-time}                            & Discrete-time        \\ \hline
Heterogeneous     & \multicolumn{1}{l|}{-}          & -                                                              & \multicolumn{1}{l|}{\cite{wang2010distributed,zuo2018adaptive}}   & \textbf{this paper}      \\ \hline
Non-heterogeneous & \multicolumn{1}{l|}{\cite{shi2009global,zou2021sampled,martinovic2022cooperative,wang2018prescribed,qiu2016distributed,ning2017distributed,ning2020bipartite,kuriki2014consensus}}          & \cite{nedic2009distributed,nedic2010constrained,nedic2014distributed} & \multicolumn{1}{l|}{\cite{zhao2017distributed,li2022exponential,wang2015distributed}} & \cite{li2011consensus} \\ \hline
\end{tabular}
}
\end{table}

\section{Preliminaries}\label{sec_pre}
\textit{Notations:} Let $\mathbb{R}^n$ be the set of vectors with dimension $n>0$. Let $\|x\|$ and $x^T$ be the standard Euclidean norm and the transpose of $x\in\mathbb{R}^n$, respectively. $I_p$ is the compatible identity matrix with dimension $p>0$ and $\otimes$ denotes the Kronecker product.
{$\col(\cdot)$ represents the column vector.}
\subsection{Graph Theory}
Following \cite{ren2005consensus}, a directed graph $\mathcal{G(V,E,A)}$ consists of  $\mathcal{V}=\{\nu_{1},\linebreak\cdots,\nu_{n}\}$ as a node set and $\mathcal{E}\in \mathcal{V}\times\mathcal{V}$ as an edge set. If the node $\nu_{i}$ is a neighbour of node $\nu_{j}$, then $(\nu_{i},\nu_{j})\in \mathcal{E}$. A directed graph is strongly connected if there exists a directed path that connects any pair of vertices. $\mathcal{A}$ is the adjacency matrix and we let $[\mathcal{A}]_{ij}=a_{ij}$ 
where $a_{ij}>0$ if $(\nu_{i},\nu_{j})\in \mathcal{E}$ and $a_{ij}=0$ otherwise. Let $\mathcal{D}$ be the degree matrix of graph $\mathcal{G(V,E,A)}$ and {$\mathcal{L=D-A}$ be the Laplacian matrix. We let $l_{ij}$ be the element of matrix $\mathcal{L}$.} 
In this paper, we assume that the information can be shared among the agents with complete information flow, formally stated in the following assumption. 

\begin{assumption}\label{asm: graph connectivity}
The communication graph is undirected and connected. 
\end{assumption}
Under this assumption, it follows that the Laplacian matrix $\mathcal L$ is semi-positive definite, and zero is a simple eigenvalue of $\mathcal L$ with an associated eigenvector $1_N$. For more detailed properties of the graph, please refer to \cite{qu2009cooperative}.




\subsection{Optimal Coordination}
This paper considers an optimal coordination problem where a group of network-connected robots are designed to solve the following optimisation problem
\begin{subequations}\label{eqn: optimal coordination problem}
\begin{gather}\small\label{eqn: 2a}
    \min_{y \in \BR^p } \ \   \sum_{i=1}^{N} f_i(y)  \\
    \text{s.t.} \ \      x_i(k+1) = A_ix_i(k) +B_iu_i(k)\label{eqn: 2b}\\
    y_i(k) = C_ix_i(k)\label{eqn: 2c}\\
    \text{for}\ \ i = 1,\cdots,N \nonumber
\end{gather}
\end{subequations}
where $N$ is the number of robots and  $x_i(k)\in \BR^n$ is the state variable of the robot $i$ that generally represents its position, speed, force and other information. $y_i(k) \in \BR^q$ is the system output of $i$th robot, which can be viewed intuitively as an observation of the robotic system since the actual state of the robot is invisible in some real robotic systems. $u_i(k) \in \BR^p$ is the system control input for robot $i$. For example, the control input can be the torque of a motor or the force of a robotic system. 
{The equations \eqref{eqn: 2b} and \eqref{eqn: 2c} represent the discrete-time linear systems, which satisfy both additivity and homogeneity \cite{chen1984linear}.} $ A_i \in \BR^{n\times n}, B_i\in \BR^{n\times p}, C_i \in \BR^{q\times n} $ are constant matrices, which indicate the \emph{heterogeneous} robotic system dynamics. 
{The linear system can be reduced to a single-integrator system if $A=I_n, B=1_{n\times p}^T$ and $C=1_{q\times n}^T$.}
$f_i(y)$ are the smooth convex function and privately known to the $i$th robot.

Under Assumption~\ref{asm: graph connectivity}, we can reformulate the optimal coordination problem. In distributed cooperative control, each individual agent will be allocated a local decision variable, denoted as $y_i$. To solve the coordination problem~\eqref{eqn: optimal coordination problem}, it is required that the local decision variables $y_i$ eventually reach to the same optimal value, i.e. $y_i=y_j, \forall i,j\in \mathcal V$. In a connected graph, it is equivalent to require $(\mathcal L\otimes I_q) Y = 0$ by noting that the null-space of $\mathcal L$ is $1_N$. 
\begin{lemma}[\cite{Li2021Automatica}]
Let Assumption~\ref{asm: graph connectivity} hold. The optimisation problem \eqref{eqn: optimal coordination problem} can be equivalently reformulated as 
\begin{equation}\small\label{problem3}
    \begin{aligned}
         \min_{y_i \in \BR^q, \forall i \in \mathcal V }  & \ \  \sum_{i=1}^{N} f_i(y_i) \\
     \text{s.t.} &  \ \  \eqref{eqn: 2b}\  \text{and}\  \eqref{eqn: 2c} \\
                & \ \ (\mathcal L\otimes I_q) Y = 0
    \end{aligned}
\end{equation}
where $Y= \col(y_1,y_2\dots, y_N)$.
\end{lemma}

{Solving problem \eqref{eqn: optimal coordination problem} is equivalent to solving the problem \eqref{problem3} when the graph is undirected and connected. Let $Y^* = \col(y_1^*,\cdots,y_N^*)$ be the optimal solution of the problem \eqref{problem3}, which means $y_i^*=y_j^*=y^* , \forall i,j \in \mathcal{V}$ by $(\mathcal L\otimes I_q) Y^* = 0$. Then, we have $\sum_{i=1}^{N} f_i(y_i^*) = \sum_{i=1}^{N} f_i(y^*)$, which implies $y^*$ is also an optimal solution to the problem \eqref{eqn: optimal coordination problem}.}

\subsection{Problem Formulation}

In many real world applications, the objective functions in \eqref{eqn: 2a} are defined according to the distance between the agent and the position of interest, for example, robotic swarm problem  \cite{soria2021nature, soria2022distributed}, source seeking \cite{park2020cooperative, ristic2020decentralised,li2021concurrent}, search and rescue \cite{azzollini2021UAV}. In this paper, we consider agent $i$ has an estimation of the target 
$r_i$. Its local objective is to track its local belief by optimising 
\begin{equation}
    f_i(y_i) = \|y_i - r_i\|^2 
\end{equation}
that is, agent $i$ intends to minimise the difference between its output and its local estimated target. 

Consequently, the problem is reformulated as 
\begin{equation}\small\label{eqn: cooperative tracking}
    \begin{aligned}
         \min_{y_i \in \mathbb R^q, \forall i \in \mathcal V }  & \ \  \sum_{i=1}^{N} f_i(y_i) = \sum_{i=1}^{N}\|y_i - r_i\|^2  \\
     \text{s.t.} &  \ \  \eqref{eqn: 2b}\  \text{and}\  \eqref{eqn: 2c} \\
                    & \ \ (\mathcal L\otimes I_q) Y = 0
    \end{aligned}
\end{equation}

{Here, we provide a concrete example of the reference $r_i$: Assuming there is a pollution source, that several robots collaborate to find. The robots have limited sensor ranges (3m). The local reference $r_i$ of robot $i$ is the highest concentration pollution point within the 3m range. Based on our algorithm, the robots will eventually converge to the pollution source point by communicating with neighbours and updating local targets.}

\begin{remark}
    The formulation in \eqref{eqn: cooperative tracking} is fundamentally different from the traditional consensus control in multi-agent systems. Consensus control \cite{ren2007information, qu2009cooperative} aims to drive the states/outputs of all agents to a consensus value that is determined by the initial values of the agents' states. The coordination problem in this paper is to operate the robot states/outputs to the optimal solutions of their joint cost functions \eqref{eqn: cooperative tracking}, which can be either static or time-varying. It is worth noting that the references $r_i, \forall i\in \mathcal V $, can be generated by learning/estimation techniques using sensory information from the local onboard sensors equipped on agent $i$. {For example, the local target could be generated based on the computer vision algorithms, such as the Yolo \cite{redmon2016you} and Apriltag \cite{wang2016apriltag,olson2011apriltag}. Here, we assume a reference point $r_i$ is given or can be detected during the operation.}
\end{remark}







\section{Algorithm Development and Convergence Analysis}\label{sec_algorithm}
\subsection{High Level Decision-Making for Single Integrator System}\label{sec: single integrator}
Currently, agents in multi-agent systems are usually devised with effective tracking algorithm to follow the instruction given by high-level decision-makers. The dynamics of agent $i$ can be expressed as a single integrator system:  
\begin{equation}\small\label{equ:single-integrator}
    y_i(k+1)=y_i(k)+v_i(k)
\end{equation}
where $y_i(k)\in\BR^q $ denotes the output of the $i$th agent. The control design can then be given by letting 
\begin{equation}\small\begin{aligned}\label{eqn: single intergrator algorithm}
     &v_i(k) =  -\beta\big[\sum_{j\in\mathcal N_i}l_{ij}y_j(k) + \sum_{j\in\mathcal N_i} l_{ij} \lambda_j(k)  +   \nabla f_i(y_i(k)) \big]\\
     & \lambda_i(k+1) =  \lambda_i (k) + \beta \sum_{j\in\mathcal N_i}l_{ij}y_j(k)
    \end{aligned}
\end{equation}
where $v_i$ denotes the control input of the $i$th agent for the integrator dynamics, $l_{ij}$ is the element of Laplacian matrix $\mathcal{L}$, $\lambda_i$ is the Lagrangian multiplier maintained by agent $i$, and $\beta$ is the optimisation gain ({equivalently the learning rate in machine learning algorithms}) to be designed. 
{We emphasise that the proposed algorithm only shares the observations $y_i(k)$ and Lagrangian multiplier $\lambda_i(k)$ among different agents, whereas the gradient term $\nabla f_i(y_i(k))$ is not transferred, and therefore the objective privacy of local agent is protected. }

 
Denote the augmented output and Lagrangian multiplier as $  Y(k) = \col(y_1(k),y_2(k), \dots, y_N(k)) $ and $\Lambda(k) = \col(\lambda_1(k),\lambda_2(k),\dots,\linebreak \lambda_N(k))$. Then, the algorithm above can be written in a compact form as 
\begin{equation}\small\begin{aligned}\label{eqn: compact single integrator dynamics}
     Y(k+1) =& Y(k) - \beta [(\mathcal L\otimes I) Y(k) + (\mathcal L\otimes I) \Lambda(k) + \nabla F(Y(k) )]\\
     \Lambda(k+1) =& \Lambda (k) + \beta(\mathcal L\otimes I) Y(k)
    \end{aligned}
\end{equation}
where {\small$\nabla F(Y(k))$ = $\col(\nabla f_1(y_1(k)), \nabla f_2(y_2(k)), \dots, \nabla f_N(y_N(k)))$}. 

As we will delineate later in the convergence analysis, iterative optimisation algorithms inherently take the form of integrators \cite{boyd2004convex}. To solve distributed optimisation problems for linear systems, we resort to output regulation techniques by taking algorithm \eqref{eqn: single intergrator algorithm} as an internal reference generator.

\subsection{Control Algorithm for Linear Multi-Robot Systems}

{The above algorithm can be extended in a nontrivial way as explained below to work with linear multi-robot systems. 
A distributed coordination algorithm to solve the optimal coordination problem is designed as follows.}
\begin{equation}\small\begin{aligned}\label{eqn: controller linear system}
	 u_i(k)= &-K_i x_i(k) +(G_i+ K_i\Psi_i) \xi_i (k) \\
	 \xi_i(k+1) =& \xi_i(k)- \beta[ \sum_{j\in \mathcal N_i} l_{ij}  \xi_j (k) + \sum_{j\in\mathcal N_i} l_{ij} \lambda_j  +   \nabla f_i(\xi_i(k))]\\
	 \lambda_i(k+1) =& \lambda_i(k) + \beta \sum_{j\in\mathcal N_i} l_{ij} \xi_j
	\end{aligned}
\end{equation}
where $\xi_i(k), \lambda_i(k)$ are two internal auxiliary variables to generate tracking reference for the $i$th agent.
$ K_i $ is chosen such that $A_i-B_iK_i$ is Schur stable under assumption that the dynamics $(A_i,B_i)$ are controllable. 
{$ G_i $ and $\Psi_i$ are gain matrices, which can be obtained by solving the following 
}
\begin{equation}\small\begin{aligned}\label{eqn: linear matrix equation}
		& (A_i-I)\Psi_i + B_iG_i= 0 \\
		& C_i\Psi_i - I = 0 
	\end{aligned}
\end{equation}
{Intuitively, the gain matrices $K_i, G_i, \Psi_i$ are designed to track the reference according to system dynamics $A_i, B_i, C_i$. For heterogeneous systems, every robot $i$ may have different $A_i, B_i, C_i$.}
To ensure the solvability of \eqref{eqn: linear matrix equation}, we adopt the following assumption \ref{asm: rank}, 
which is a regulation equation in the output regulation literature \cite{huang2004nonlinear}. 
\begin{assumption}\label{asm: rank}
	The pairs $(A_i,B_i),\forall i \in \mathcal V$ are controllable,
and 
\begin{equation}\small\label{eqn: rank condition}
	\operatorname{rank}\left[\begin{array}{cc}
		A_i-I & B_i\\
			     C_i & 0 
	\end{array}\right]=n +q .
\end{equation}
\end{assumption}

\begin{remark}
The proposed algorithm is in fact a combination of gradient-descent optimisation and output regulation techniques. The internal model $\xi_i(k)$ is generated by consensus based gradient-descent optimisation with $\lambda_i (k) $ being the Lagrangian multiplier. {The design of the control input $u_i$ is motivated by the classic output regulation approach \cite{huang2004nonlinear}.}
\end{remark}

Similarly, the closed-loop system dynamics can be compactly written as 
\begin{equation}\small\begin{aligned}\label{eqn: controller linear system2}
    U(k) = &-K X(k) + \Pi\Xi(k) \\
     \Xi(k+1) =& \Xi(k) - \beta [(\mathcal L\otimes I) \Xi(k) + (\mathcal L\otimes I) \Lambda(k) +\ \nabla F(\Xi(k)) ]\\
     \Lambda(k+1) =& \Lambda (k) + \beta(\mathcal L\otimes I) Y(k)
    \end{aligned}
\end{equation}
where {$K = \diag(K_1, \dots, K_N)$ and $\Pi =\diag (G_1+K_1\Psi_1, \dots, G_N+K_N\Psi_N)$}.

\subsection{Convergence Analysis}

{The convergence analysis of the proposed algorithm proceeds in three steps. In the first step \ref{parta}, we prove that the equilibrium point of the proposed algorithm is the optimal solution of the problem \eqref{eqn: cooperative tracking}. After that, the proposed algorithm can guarantee that both single integrator and linear system will converge to the equilibrium point, which are proven in step \ref{partb} and \ref{partc}, respectively.}
\subsubsection{}\label{parta}We begin with the convergence analysis for high-level decision making in Section~\ref{sec: single integrator}, which will then serve as a reference generator later in the proof of linear system regulation. 

In \eqref{eqn: compact single integrator dynamics}, the equilibrium point, denoted as $(Y^*,\Lambda^*)$, satisfy
\begin{equation}\small\begin{aligned}\label{eqn: equilibrium single integrator}
         \beta [(\mathcal L\otimes I) Y^* -  (\mathcal L\otimes I) \Lambda^* - \nabla F(Y^*) ] = 0\\
        \beta (\mathcal L\otimes I) Y^*= 0.
    \end{aligned}
\end{equation}
Invoking the properties of $\mathcal L$ under Assumption~\ref{asm: graph connectivity}, \eqref{eqn: equilibrium single integrator} yields 
\begin{equation}\small\label{eqn: zero gradient}
    1_N^T \nabla F(Y^*)= 0
\end{equation}
Since $(\mathcal L\otimes I) Y^*= 0$ implies $y_i=y_j = y^*, \forall i,j\in \mathcal V$, it follows from \eqref{eqn: zero gradient} that 
\begin{equation}\small
     \sum_{i=1}^{N} \nabla f_i(y^*) =0.
\end{equation}
Note that the uniqueness of $y^*$ is guaranteed as the local objective functions are all strongly convex. According to the primal-dual theory \cite{lei2016primal}, the solution pair $(Y^*,\Lambda^*)$ is in fact a saddle point of the Lagrangian function of $\phi(Y, \Lambda)=\sum_{i=1}^N f_i(y_i)+\Lambda^T(\mathcal{L} \otimes I_{q}) Y$.

Now, we are ready to give the convergence of \eqref{eqn: compact single integrator dynamics} to the equilibrium $(Y^*,\Lambda^*)$ for single integrator dynamics. {It is noticed that the equilibrium points of algorithm \eqref{eqn: single intergrator algorithm} and \eqref{eqn: controller linear system} are same.}

\subsubsection{}\label{partb}
{To give a complete proof, we next show that the proposed algorithm can make the system asymptotically converge to the equilibrium point. }

\begin{theorem}\label{thm: 1}
Let Assumption~\ref{asm: graph connectivity} hold. If the step size $\beta$ is chosen according to $0<\beta<\min \{\frac{1}{2\lambda_{max}(\mathcal L)}, \frac{3}{2L}\}$ with $\lambda_{max}(\mathcal L)$ being the maximum eigenvalue of the Laplacian matrix and $L$ being the Lipschitz constant of the cost function, then algorithm~\eqref{eqn: compact single integrator dynamics} converges to the optimal solution $(Y^*,\Lambda^*)$.
\end{theorem}
\begin{proof}
From the second update Equation in~\eqref{eqn: compact single integrator dynamics}, we have 
\begin{equation}\small
\begin{aligned}\label{eqn: 17}
\Lambda(k)-\Lambda^{*}=& \Lambda(k+2)-\Lambda^{*}-\beta(\mathcal{L} \otimes I)(Y(k)-Y^{*}) \\
&-\beta(\mathcal{L} \otimes I)(Y(k+1)-Y^{*})
\end{aligned}
\end{equation}
Left-multiplying both sides of \eqref{eqn: 17} yields 
\begin{equation}\small\begin{aligned}\label{eqn: 18}
   (\mathcal{L} \otimes I)( \Lambda(k)-\Lambda^{*})=& (\mathcal{L} \otimes I)(\Lambda(k+2)-\Lambda^{*}) \\ & -\beta(\mathcal{L}^2 \otimes I)(Y(k)-Y^{*}) \\
&-\beta(\mathcal{L}^2 \otimes I)(Y(k+1)-Y^{*})
\end{aligned}
\end{equation}
where we have used the property of Kronecker product $(A \otimes B)(C \otimes D)=(A C) \otimes(B D)$. 
Denote $Z(k+1)=Y(k)-\beta \nabla F(Y(k))-\beta(\mathcal{L} \otimes I)(\Lambda(k)+Y(k))-Y(k+1)$ and $\mathcal{W}=(I-\beta \mathcal{L}+\beta^{2} \mathcal{L}^{2}) \otimes I$. Combining \eqref{eqn: 18} and \eqref{eqn: compact single integrator dynamics}, we have:
\begin{equation}\small\begin{aligned}\label{eqn: 19}
   & Y(k+1) - Y^* \\=& Y(k)-Y^* - \beta\nabla F(Y(k)) - \beta (\mathcal{L}\otimes I)(\Lambda(k)+Y(k)) - Z(k+1)
   \\=&((I-\beta\mathcal{L}+\beta^2\mathcal{L}^2)\otimes I)(Y(k)-Y^*)-Z(k+1)\\&-\beta((\mathcal{L}\otimes I)(\Lambda(k+2)-\Lambda^*)- \beta(\nabla F(Y(k))-\nabla F(Y^*))\\&-\beta(\nabla F(Y^*)+(\mathcal{L}\otimes I)\Lambda^*)+\beta^2(\mathcal{L}^2\otimes I)(Y(k+1)-Y^*)
\end{aligned}
\end{equation}

Substracting $\mathcal{W} (Y(k+1)-Y^*)$ from both sides of Equation \eqref{eqn: 19}, a new recursive update law is written as 
\begin{equation}\small\label{eqn: W}\small
    \begin{aligned} \mathcal{W} &(Y(k+1)-Y(k))+\beta(\mathcal{L} \otimes I)(\Lambda(k+2)-\Lambda^{*}) \\=&-(\beta \mathcal{L} \otimes I-2 \beta^{2} \mathcal{L}^{2} \otimes I)(Y(k+1)-Y^{*})-Z(k+1) \\ &-\beta(\nabla F(Y(k))-\nabla F(Y^{*}))-\beta(\nabla F(Y^{*})+(\mathcal{L} \otimes I) \Lambda^{*}). \end{aligned}
\end{equation}
With this new recursive algorithm, the convergence can be established using Lyapunov method and saddle point dynamics. 

We formulate a Lyapunov function as follows:
\begin{equation}\small
    V(Y,\Lambda) = \langle Y-Y^*,\mathcal{W} (Y-Y^*)\rangle + \langle \Lambda-\Lambda^*,\Lambda-\Lambda^*\rangle
\end{equation}
Here, we use $\langle x, y\rangle$ to represent the inner product of vectors $x$ and $y$. 
Then we have:
\begin{equation}\small\begin{aligned}\label{eqn: 21}
&V(Y(k+1),\Lambda(k+2)) - V(Y(k),\Lambda(k+1))\\
=&\langle Y(k+1)-Y^*,\mathcal{W} (Y(k+1)-Y^*)\rangle + \|\Lambda(k+2)-\Lambda^*\|^2\\
&-\langle Y(k)-Y^*,\mathcal{W} (Y(k)-Y^*)\rangle + \|\Lambda(k+1)-\Lambda^*\|^2\\
=&-\langle Y(k+1)-Y(k),\mathcal{W} (Y(k+1)-Y(k))\rangle \\&+ 2\langle Y(k+1)-Y^*,\mathcal{W} (Y(k+1)-Y(k))\rangle\\&-\|\Lambda(k+2)-\Lambda(k+1)\|^2\\&+2\langle \Lambda(k+2)-\Lambda(k+1),\Lambda(k+2)-Y^*\rangle
\end{aligned}
\end{equation}

Based on the Equations \eqref{eqn: controller linear system2} and \eqref{eqn: equilibrium single integrator}, we can derive the last term of Equation \eqref{eqn: 21} as:
\begin{equation}\small\begin{aligned}\label{eqn: 22}
&\langle \Lambda(k+2)-\Lambda(k+1),\Lambda(k+2)-Y^*\rangle \\=&\langle \beta(\mathcal{L}\otimes I)(Y(k+1)-Y^*),\Lambda(k+2)-\Lambda^*\rangle
\\=&\langle \beta(\mathcal{L}\otimes I)(\Lambda(k+2)-\Lambda^*),Y(k+1)-Y^*\rangle
\end{aligned}
\end{equation}

Therefore, we can obtain:
\begin{equation}\small\begin{aligned}\label{eqn: 23}
&\langle Y(k+1)-Y^*,\mathcal{W} (Y(k+1)-Y(k))\rangle\\&+\langle \Lambda(k+2)-\Lambda(k+1),\Lambda(k+2)-Y^*\rangle\\ = &\langle Y(k+1)-Y^*, \beta(\mathcal{L}\otimes I)(\Lambda(k+2)-\Lambda^*)\\&+\mathcal{W}(Y(k+1)-Y(k))\rangle
\end{aligned}
\end{equation}

With Equation \eqref{eqn: W}, we can further derive as:
\begin{equation}
\small\begin{aligned}\label{eqn: 24}
&\langle Y(k+1)-Y^*,\mathcal{W} (Y(k+1)-Y(k))\rangle\\&+\langle \Lambda(k+2)-\Lambda(k+1),\Lambda(k+2)-Y^*\rangle\\
=&-\langle Y(k+1)-Y^*,(\beta \mathcal{L} \otimes I-2 \beta^{2} \mathcal{L}^{2} \otimes I)(Y(k+1)-Y^{*})\rangle\\&-\langle Y(k+1)-Y^*,\beta(\nabla F(Y(k))-\nabla F(Y^{*}))\rangle\\&-\langle \Lambda(k+2)-\Lambda(k+1),Z(k+1)\rangle \\&-\langle Y(k+1)-Y^*,\beta(\nabla F(Y^{*})+(\mathcal{L} \otimes I) \Lambda^{*})\rangle
\end{aligned}
\end{equation}

Since $f(y)$ is smoothly convex function, then we can derive based on the optimal condition \cite{bertsekas2009convex}:
\begin{equation}\small\label{eqn:25}
    \langle Y(k+1)-Y^*,\beta(\nabla F(Y^{*})+(\mathcal{L} \otimes I) \Lambda^{*})\rangle\geq 0
\end{equation}

By the definition of normal cone, we have:
\begin{equation}\small\label{eqn:26}
    \langle Y(k+1)-Y^*,Z(k+1)\rangle\geq 0
\end{equation}

Substituting \eqref{eqn: 24}, \eqref{eqn:25} and \eqref{eqn:26} back to \eqref{eqn: 21}, we obtain:
\begin{equation}
\small\begin{aligned}\label{eqn: 27}
&V(Y(k+1),\Lambda(k+2)) - V(Y(k),\Lambda(k+1))\\
\leq &-\langle Y(k+1)-Y(k),\mathcal{W} (Y(k+1)-Y(k))\rangle \\&-\|\Lambda(k+2)-\Lambda(k+1)\|^2\\&-2\langle Y(k+1)-Y^*,(\beta \mathcal{L} \otimes I-2 \beta^{2} \mathcal{L}^{2} \otimes I)(Y(k+1)-Y^{*})\rangle\\&-2\langle Y(k+1)-Y^*,\beta(\nabla F(Y(k))-\nabla F(Y^{*}))\rangle
\end{aligned}
\end{equation}

From Theorem \ref{thm: 1}, we have $0<\beta\leq\frac{1}{2\lambda_{max}(\mathcal{L})}$. Moreover, $\mathcal{L}$ is symmetric with a zero eigenvalue, and therefore, we could find an orthogonal matrix $\mathcal{P}$ that $\mathcal{P}^T\mathcal{L}\mathcal{P} = diag\{0,\lambda_1,\cdots,\lambda_N\}$ and $\mathcal{P}^T\mathcal{L}^2\mathcal{P} = diag\{0,\lambda_1^2,\cdots,\lambda_N^2\}$. Then, the matrix $\beta \mathcal{L} -2 \beta^{2} \mathcal{L}^{2}$ is positive semi-definite.

Furthermore, the cost function is Lipschitz continuous:
\begin{equation}\small\label{eqn: 28}
    \langle Y-Y^*,\nabla F(Y)-\nabla F(Y^{*}))\geq \frac{1}{L}\|\nabla F(Y)-\nabla F(Y^{*})\|^2
\end{equation}
Applying Jensen's inequality to the last term of \eqref{eqn: 27}:
\begin{equation}
    \small\begin{aligned}\label{eqn: 29}
&-\langle Y(k+1)-Y^*,\nabla F(Y(k))-\nabla F(Y^{*}))\rangle\\ = &-\langle Y(k)-Y^*,\nabla F(Y(k))-\nabla F(Y^{*}))\\&+\langle -Y(k+1)+Y(k),\nabla F(Y(k))-\nabla F(Y^{*}))\\\leq&-\frac{1}{L}\|\nabla F(Y(k))-\nabla F(Y^{*})\|^2 + \frac{L}{4}\|Y(k)-Y(k+1)\|^2\\&+\frac{1}{L}\|\nabla F(Y(k))-\nabla F(Y^{*})\|^2\\\leq&\frac{L}{4}\|Y(k)-Y(k+1)\|^2
\end{aligned}
\end{equation}

Then the Equation \eqref{eqn: 27} can be reformed as:
\begin{equation}
\small\begin{aligned}\label{eqn: 30}
&V(Y(k+1),\Lambda(k+2)) - V(Y(k),\Lambda(k+1))\\
\leq &-\langle Y(k+1)-Y(k),(\mathcal{W}-\frac{\beta L}{2}I) (Y(k+1)-Y(k))\rangle \\&-\|\Lambda(k+2)-\Lambda(k+1)\|^2
\end{aligned}
\end{equation}

With Theorem \ref{thm: 1}, we have $0< \beta \leq \frac{3}{2L}$, which means $1-\beta\lambda+\beta^2\lambda^2-\frac{\beta L}{2} = (\frac{1}{2}-\beta\lambda)^2+\frac{3}{4}-\frac{\beta L}{2}>0$. Therefore, $(\mathcal{W}-\frac{\beta L}{2}I)$ is positive definite. In consequence, $V(Y(k+1),\Lambda(k+2)) \leq V(Y(k),\Lambda(k+1))$. Thus, we can conclude that $V(Y,\Lambda)$ converges under the condition of Theorem \ref{thm: 1}.
\end{proof}
\subsubsection{}\label{partc}Before presenting the main result for the distributed optimal cooperative control for general linear systems, we need to apply a state transformation to \eqref{eqn: controller linear system} by letting $x_{i,s}(k) = \Psi_i \xi_i(k) $, $u_{i,s}(k) = G_i \xi_i(k)$. Let $\bar x_i(k) = x_i(k)-x_{i,s}(k)$ and $\bar u_i(k) = u_i(k)-u_{i,s}(k)$. Applying the control input \eqref{eqn: controller linear system}, we have the closed-loop dynamics 
\begin{equation}\small\label{eqn: closed-loop error dynamics}\begin{aligned}
	\bar x_i(k+1)  & =  (A_i-B_iK_i)\bar x_i(k) - \Psi_i v_i(k) \\
			e_i(k) & = C_i\bar x_i(k). 
\end{aligned}
\end{equation}

The following lemma can be obtained, which can be regarded as input-to-output stability. 

\begin{lemma}\label{lem: input-output stability}
Let Assumptions~\ref{asm: graph connectivity} and \ref{asm: rank} hold. If $K_i$ is chosen such that $A_i-B_iK_i$ is Schur stable and the $G_i$ and $\Psi_i$ are designed by solving the regulation Equations in \eqref{eqn: linear matrix equation}, then there exists a positive constant $\alpha>0$ such that 
\begin{equation}\small\label{eqn: input to output stability}
		\limsup_{k\rightarrow \infty} \| e_i(k) \| \leq  \alpha  \limsup_{k\rightarrow \infty} \|v_i(k)\|.
	\end{equation}
\end{lemma}

\begin{proof}
First, we examine the solvability of \eqref{eqn: linear matrix equation}. 
	Encapsulating \eqref{eqn: linear matrix equation} into a matrix form leads to 
	\begin{equation}\small\label{eqn: matrix equation}
		\left[\begin{array}{cc}
			A_i-I & B_i\\
				     C_i & 0
		\end{array} \right]   \left[\begin{array}{c}
			\Psi_i  \\
			G_i
		\end{array} \right]=  \left[\begin{array}{c}
			0 \\
			I
		\end{array} \right].
	\end{equation}
	Denoting $\mathcal O_i = 	\left[\begin{array}{cc}
		A_i-I & B_i\\
			     C_i & 0 
	\end{array} \right]  $ and $T_i = \left[\begin{array}{c}
		\Psi_i \\
		G_i
	\end{array} \right]$, by leveraging the property of Kronecker product, $\operatorname{vec}\left(\mathcal O_i T_i I \right)=\left(I \otimes \mathcal O_i\right) \operatorname{vec}(T_i)$, we can obtain a standard linear algebraic equation 
	\begin{equation}\small\label{key}
		\left(I \otimes \mathcal O_i\right) \operatorname{vec}(T_i) =\operatorname{vec} \left(\left[\begin{array}{c}
			0  \\
			I
		\end{array} \right]  \right) 
	\end{equation}
of which the solvability is guaranteed under \eqref{eqn: rank condition} in Assumption~\ref{asm: rank}. 
	
	For notational convenience, we denote $A_{i,c} = A_i-B_iK_i$ and $B_{i,c} = -\Psi_i$. Then, we have 
\begin{equation}\small\label{eqn: closed-loop error dynamics 2}\begin{aligned}
		\bar x_i(k+1)  & =  A_{i,c}\bar x_i (k) +B_{i,c} v_i (k) .
	\end{aligned}
\end{equation}
Recursively iterating \eqref{eqn: closed-loop error dynamics 2} results in 
\begin{equation}\small\label{eqn: denition of time index}
	\bar x_i(k) = A_{i,c}^k\bar x_i(0) +\sum_{j=0}^{k-1} A_{i,c}^{k-j-1}B_{i,c} v_i(j).
\end{equation}
Hence, we have
\begin{equation}\small\label{eqn: error equation 1}
	e_i(k) = C_i\bar x(k) =C_iA_{i,c}^k\bar x_i(0) - \sum_{j=0}^{k-1} A_{i,c}^{k-j-1} v_i(j)
\end{equation}
where $C_i\Psi_i - I =0 $ has been used. Because $A_{i,c} $ is Schur, we have $\lim_{k\rightarrow \infty} C_iA_{i,c}^k\bar x_i(0) = 0 $. In Theorem~\ref{thm: 1}, the convergence of reference generator has been established, which implies $v_i$ converges to zero as $k\rightarrow \infty$. Denote $\varpi_i:= \limsup_{k\rightarrow \infty} \| v_i(k)\| $. Then, for any small constant $\epsilon>0$, there exists a positive time index $\zeta>0$ such that 
\begin{equation}\small\label{eqn: bounds of gradient term}
	\| v_i(k)\| < \varpi_i+\epsilon, \  \forall k>\zeta. 
\end{equation}
Based on the time index $\zeta$, the second term in \eqref{eqn: error equation 1} can be separated into two parts, written as 
\begin{equation}\small\label{eqn: second term separation}
	\sum_{j=0}^{k-1} A_{i,c}^{k-j-1} v_i(j) = \sum_{j=0}^{\zeta} A_{i,c}^{k-j-1} v_i(j)+\sum_{j=\zeta+1}^{k-1} A_{i,c}^{k-j-1} v_i(j).
\end{equation}
Taking the Euclidean norm of \eqref{eqn: second term separation} and invoking \eqref{eqn: bounds of gradient term}:
\begin{equation}\small\label{eqn: bounded second term}\begin{aligned}
	&\bigg\|\sum_{j=0}^{k-1} A_{i,c}^{k-j-1} v_i(j)\bigg\| =  \bigg\| \sum_{j=0}^{\zeta} A_{i,c}^{k-j-1} v_i (j) +\sum_{j=\zeta+1}^{k-1} A_{i,c}^{k-j-1} v_i(j) \bigg\| \\
	& \leq \big\| A_{i,c}^{k-\zeta-1} \big \| \bigg\| \sum_{j=0}^{\zeta} A_{i,c}^{\zeta-j} v_i(j) \bigg \| + (\varpi_i +\epsilon) \bigg\| \sum_{j=\zeta+1}^{k-1} A_{i,c}^{k-j-1}  \bigg\| . 
	\end{aligned}
\end{equation}
Therefore, we have 
\begin{equation}\small\label{eqn: bound of e}
	\limsup_{k \rightarrow \infty} \|e_i(k) \| \leq \frac{1}{1-\|A_{i,c}\|}\left(\varpi_i +\varepsilon\right)
\end{equation}
where the following two results have been applied
\begin{equation}\small
	\sum_{j=\zeta+1}^{t-1}\|A_{i,c}\|^{t-1-j}=\frac{1-\|A_{i,c}\|^{t-\zeta}}{1-\|A_{i,c}\|}<\frac{1}{1-\|A_{i,c}\|}
\end{equation}
\begin{equation}\small
	\lim_{k\rightarrow \infty} \big\| A_{i,c}^{k-\zeta-1} \big \|  = 0 .
\end{equation}
As $\epsilon $ can be set arbitrarily small, it follows from \eqref{eqn: bound of e} that 
\begin{equation}\small\label{eqn: input to output stability 2}
		\limsup_{k\rightarrow \infty} \| e_i(k) \| \leq \alpha \limsup_{k\rightarrow \infty} \|v_i(k)\|.
\end{equation}
where $\alpha =\frac{1}{1-\|A_{i,c}\|} $. 
\end{proof}

\begin{theorem}
Let Assumptions~\ref{asm: graph connectivity} and \ref{asm: rank} hold. If $K_i$ are chosen such that $A_i-B_iK_i$ is Schur stable and the $G_i$ and $\Psi_i$ are designed by solving the regulation equations in \eqref{eqn: linear matrix equation}, then the proposed algorithm in \eqref{eqn: controller linear system} solves the optimal coordination problem in \eqref{eqn: optimal coordination problem}.
\end{theorem}
\begin{proof}
Let $\tilde x_i(k) = x_i(k)-\Psi_i y^*$,  we have 
\begin{equation}\small\begin{aligned} 
	\tilde x_i(k+1) = & A_ix_i(k) + B_i [-K_ix_i(k) +(G_i+ K_i\Psi_i) \xi_i (k) ] - \Psi_i y^* \\
	 = & (A_i-B_iK_i) \tilde x_i(k) + B_i(G_i+K_i\Psi_i) (\xi_i(k) - y^*).
	\end{aligned}
\end{equation}
It follows from Theorem~\ref{thm: 1} and Lemma~\ref{lem: input-output stability} that $\xi_i(k) $ converges to $y^*$. Thus, we can conclude the convergence of the proposed algorithm \eqref{eqn: controller linear system} by treating $B_i(G_i+K_i\Psi_i) (\xi_i(k) - y^*)$ as $v_i(k)$ in Lemma~\ref{lem: input-output stability}. 
\end{proof}

It is worth mentioning that recursive updating algorithms usually have the format of $    y_i(k+1)=y_i(k)+v_i(k)$, where $v_i$ is the change of $y_i(k)$ at time step $k$. In this paper, we started with this type of integrator dynamics, and then we extended our algorithm to general linear systems by using the classic internal model approach in output regulation where the reference generator has the same updating mechanism as a single integrator.  

\begin{remark}
In view of the distributed cooperative control literature, e.g., \cite{qu2009cooperative}, the classic consensus control problem can be regarded as a special case of the optimal cooperative optimisation problems studied in this paper by setting the cost functions as $f_i(y) = (y-y_i(0))^2$. The research problem considered in this paper is initialisation-independent in the sense that the global optimal solution is obtained by solving a distributed optimal coordination problem. 
\end{remark}

\begin{remark}
    {In this remark, we discuss the scalability of the solutions. The proposed algorithm is gradient-based, which is easy to compute. The agents only need to communicate with their neighbours. The communication complexity for every agent is $O(n|E|)$, where E is the set of communication edges and n is the number of iterations, because in every iteration the agent only needs to send one message to their neighbours and every message is of constant size $O(1)$. The computational complexity for every agent is also $O(n|E|)$, since every iteration the agent can process the received messages and update the local state in a linear way.}
\end{remark}

\section{Experiment Results}\label{sec_sim}
This section presents the experimental results to evaluate the effectiveness of the proposed distributed optimisation algorithm. In the beginning, numerical case studies are presented to test the proposed algorithm on multi-robots with linear systems. Then, the proposed algorithm is validated and applied on real Turtlebot robots, where all source code, distributed algorithm details and environment setting files are publicly available at our project website
\footnote{Project website: \url{https://github.com/YD-19/DO4.git}}.
\subsection{Numerical Simulation}
We test the algorithm on two multi-robot systems consisting of four and ten robots, respectively. Each robot only communicates with its neighbouring robot, and the connected graphs are shown in Fig. \ref{graph-case 1}. For demonstration, we implement the algorithm for a group of four agents (case A), and then extend it to a network of ten agents for the scalability test of the proposed algorithm (case B). 

\begin{figure}[htbp]
    \centering
    \includegraphics[width=0.7\hsize]{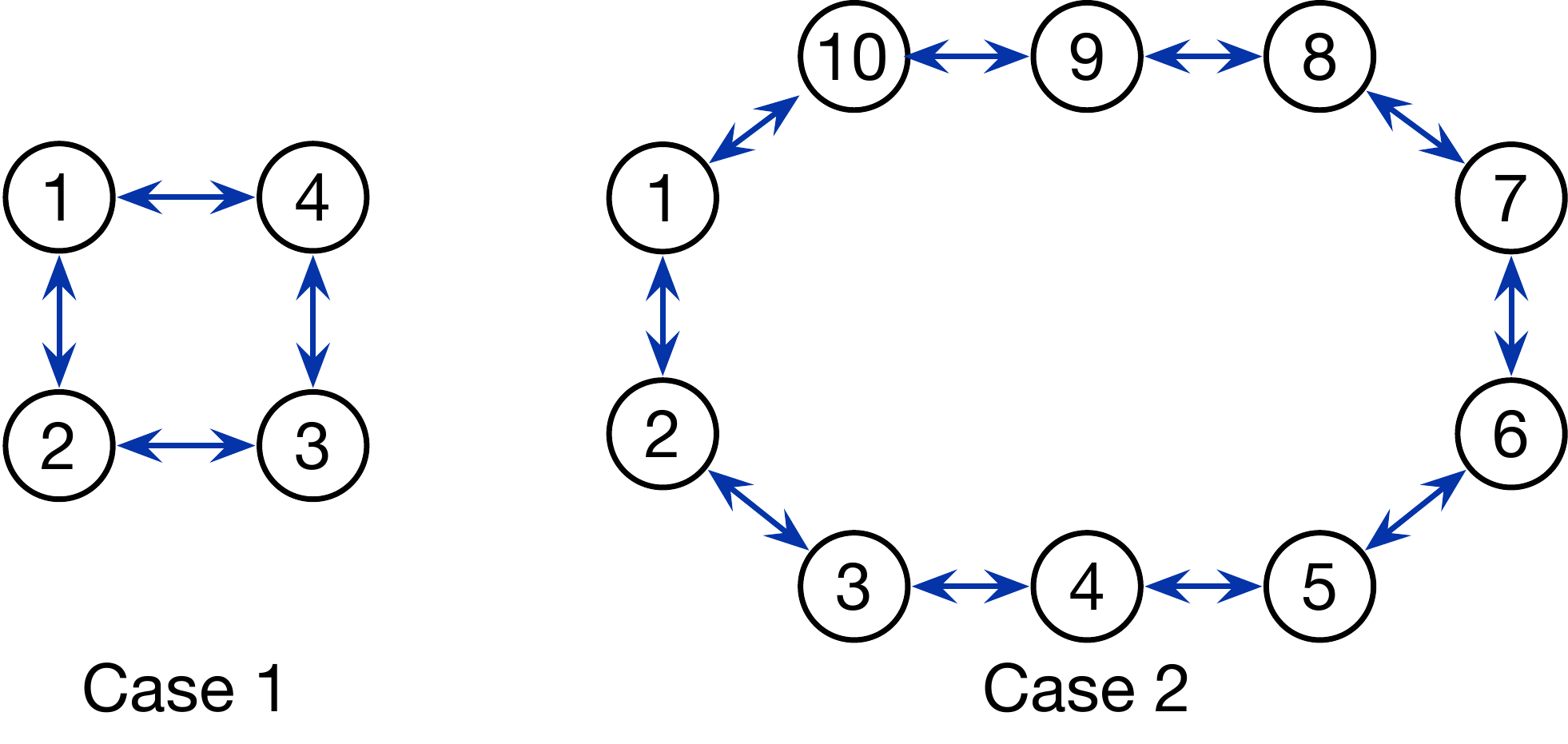}
    \caption{Ring communication topologies for a group of four and ten agents, respectively.}
    \label{graph-case 1}
\end{figure}
\paragraph{Case A}
The agents dynamics are specified as $A=[0, 1; 2, 1], B=[1, 0; 0, 1]$ and $C=[1, 0; 0, 1]$.
The assumptions on the graph connectivity, controllablity and regulation conditions are all met. The references of the agents are set as $r_1=[10;1],r_2=[5;10];r_3=[10;2];r_4=[3;5] $. 
\begin{figure}[htbp]
    \centering
    \includegraphics[width=\hsize]{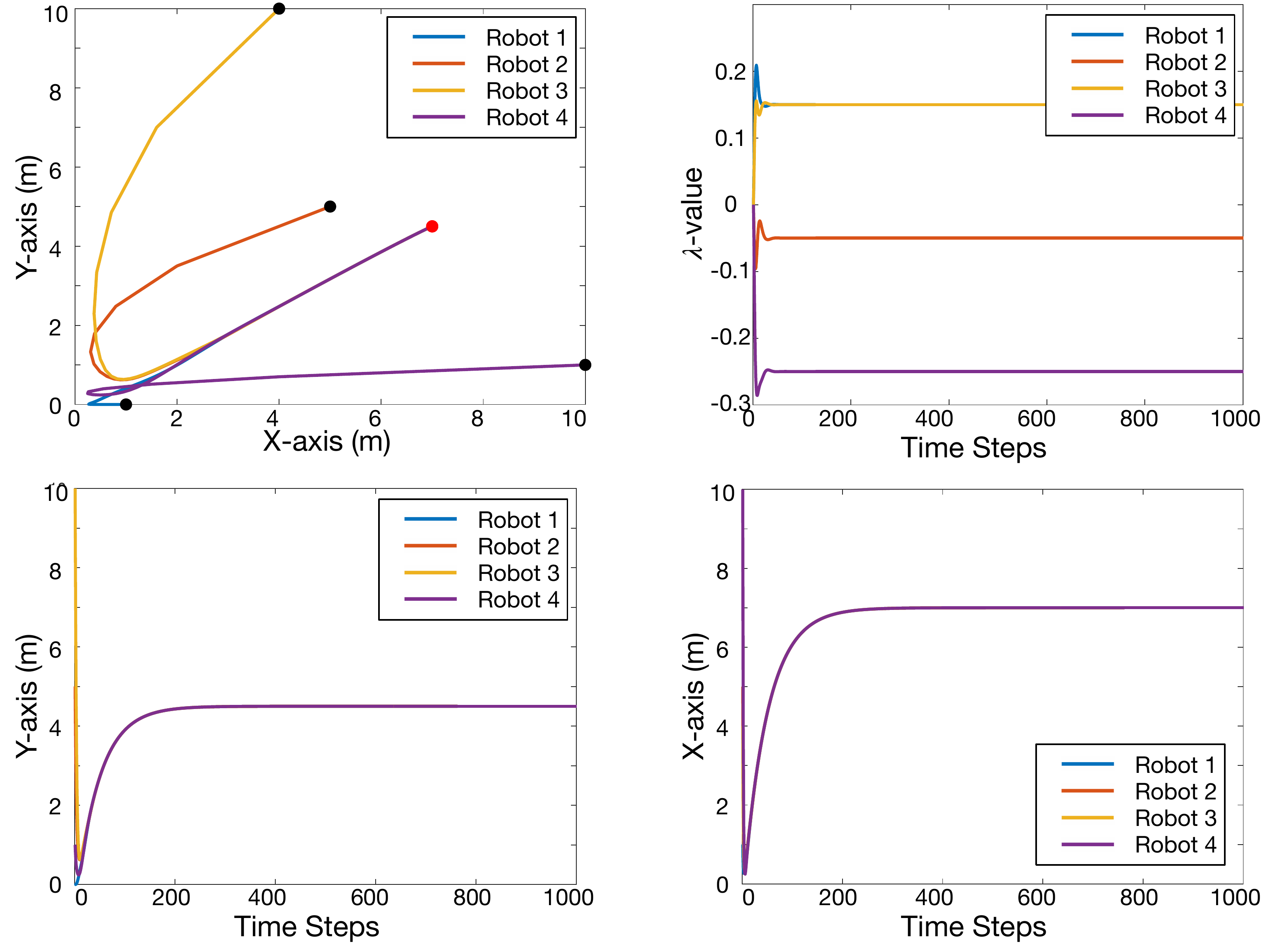}
    \caption{Simulation results for a group of four agents with static local tracking references.}
    \label{fig_caseA}
\end{figure}

The simulation results are shown in Fig. \ref{fig_caseA}. The top left figure is the top view of agents' trajectories. The black points are the initial positions of the four agents, and red points are the final positions of four agents. It can be seen that the agents are driven to the same optimal location which is independent of their initial states but related to the optimisation solution of the references specified. To illustrate more clearly, the convergence processes are shown in the rest three sub-figures, where the top right figure shows the convergence of the intermediate variable $\lambda_i$ and the bottom figures show the details of the position states $x_i$ and $y_i$, respectively.
From Fig. \ref{fig_caseA}, we can see that the agents reach the consistent goal and the parameters are converged. 
\begin{figure}[htbp]
    \centering
    \includegraphics[width=\hsize]{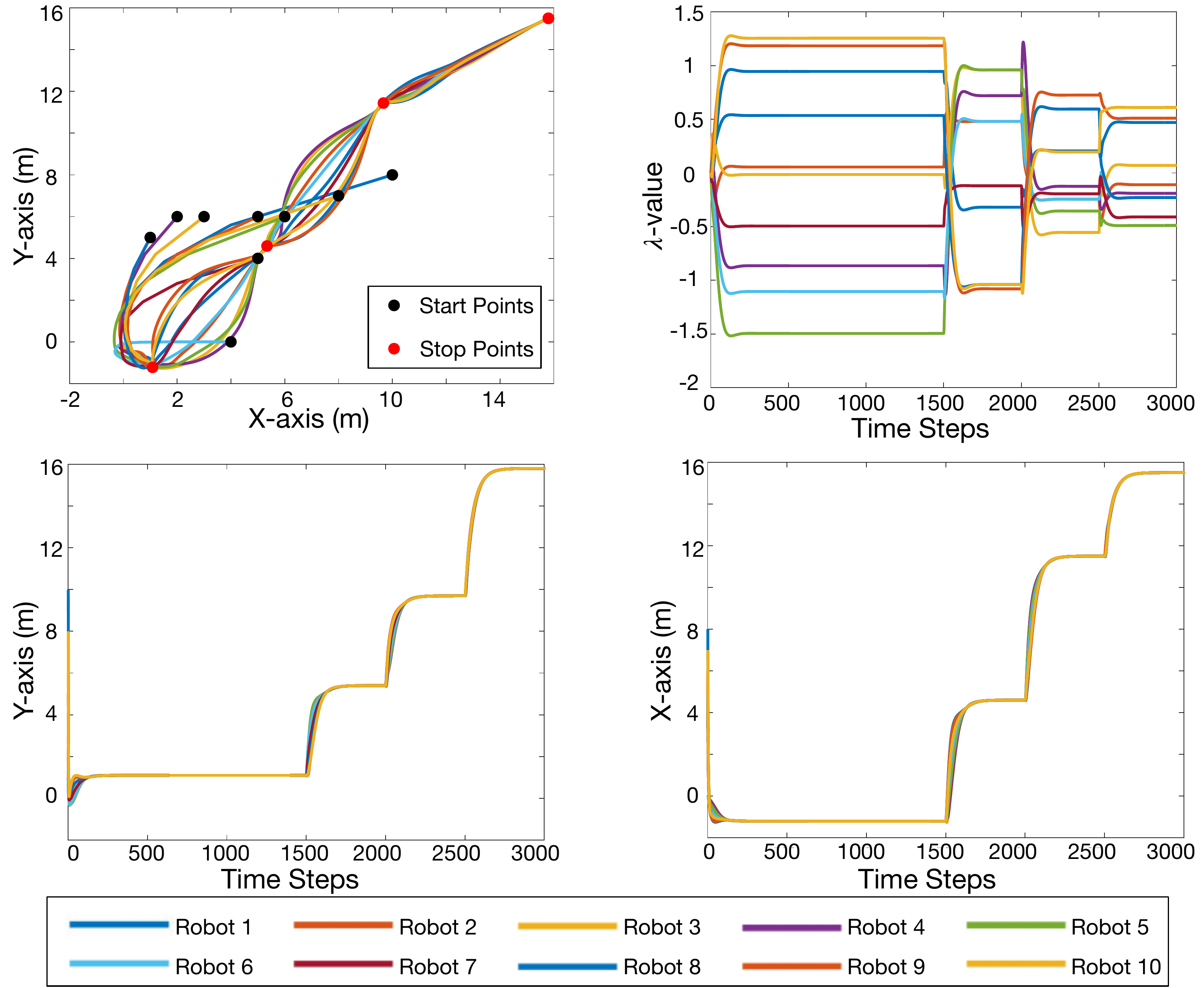}
    \caption{Simulation results for a group of ten agents with rescheduled tracking targets.}
    \label{fig_caseB}
\end{figure}

\begin{figure}[htbp]
    \centering
    \includegraphics[width=0.5\hsize]{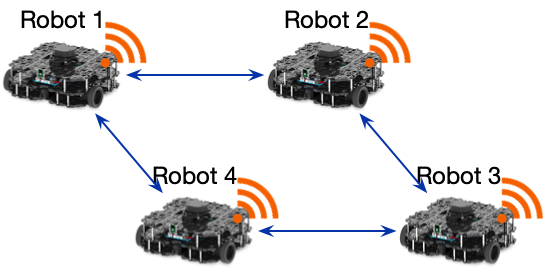}
    \caption{Communication graph for the Turtlebot network.}
    \label{graph}
\end{figure}

\begin{figure*}[h]
    \centering
    \includegraphics[width=\hsize]{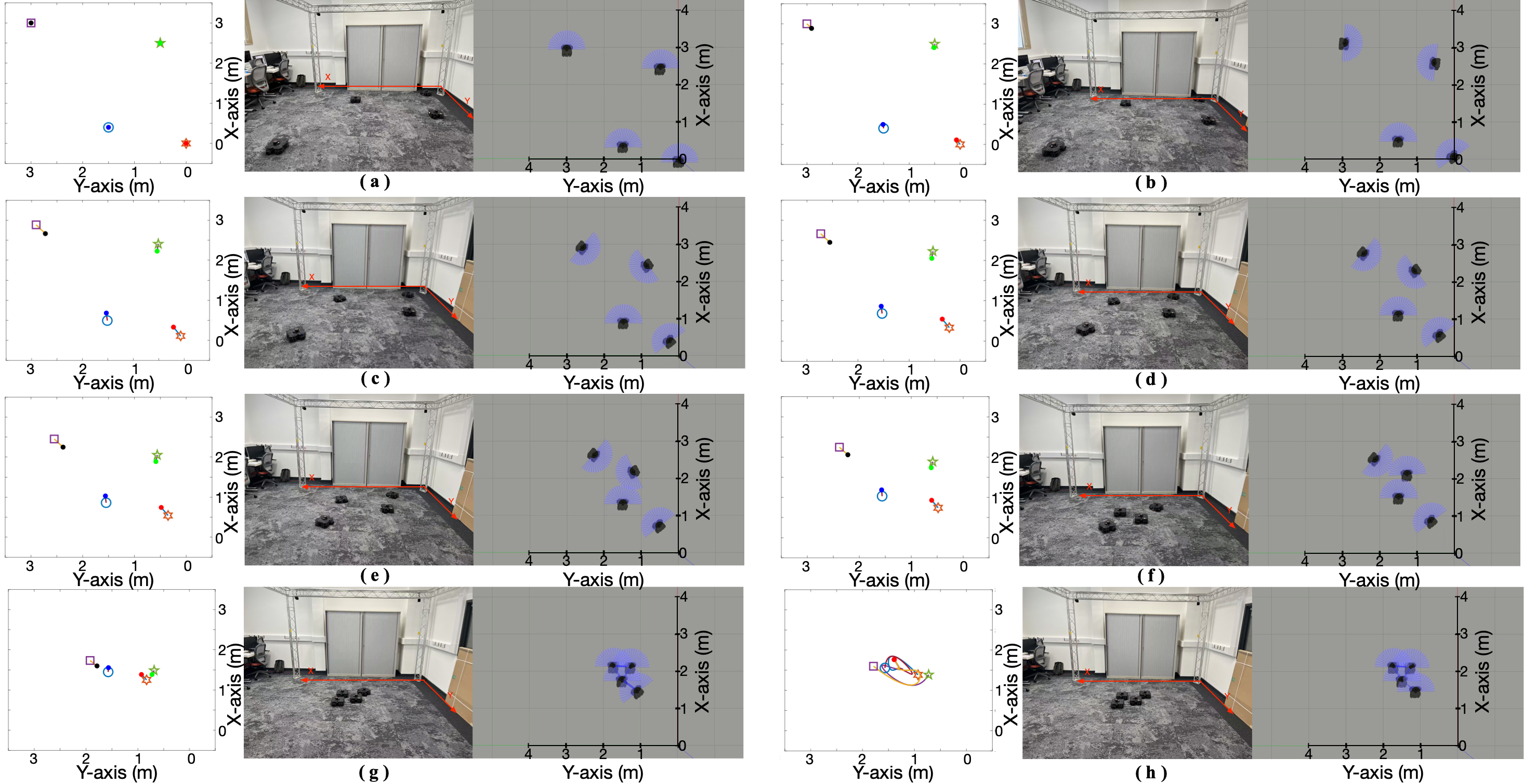}
    \caption{Experiment results on real multi-robot systems.}
    \label{fig3}
\end{figure*}

\paragraph{Case B}

We further examine the scalability of the proposed algorithm by implementing it for a set of ten agents, where each agent only needs to communicate and cooperate with its neighbours. This implies that expanding the size of the network does not incur any additional communication and computation burden, which is one of the advantageous features of distributed methods. 

In addition, we consider a scenario in which the robots' targets change over time. For example, when the robots are controlled to the first target position based on the initial observations, they may re-scan the environment and make a new goal orientation. Therefore, we assume that the robots would re-observe their targets at $1500$, $2000$, and $2500$ steps.

From the simulation results in Fig. \ref{fig_caseB}, it is observed that the proposed algorithm can optimally track the global optimal target no matter where and when the agents re-position their local targets. Therefore, this case study shows both the scalability and initialisation-independent properties of the proposed algorithm.

\subsection{Real Multi-Robot Systems}
In this section, we apply the proposed distributed algorithm to the simulation and the physical environments. The robot simulation is based on the ROS and Gazebo platforms. For the physical environment, we assemble four Turtlebot3 Waffle Pi robots, together with a laboratory environment, as shown in the middle of sub-figures in Fig. \ref{fig3}. The Turtlebot3 robots can only get information and communicate with their neighbours, with the communication graph shown in Fig. \ref{graph}. They are all driven based on the standard ROS platform. {We applied the distributed optimisation algorithm to generate an optimal position for every time point. The underlying PID controller will drive the robot to follow the optimal position based on the error between the current and optimal positions.} After the top-level optimisation algorithm publishes an optimised reference, the low-level PID controller will trigger two motors to track the reference signal based on Raspberry Pi and OpenCR. 

In this experiment, the environment is limited to the laboratory room size, and therefore we design the simulation and real environment within a $3.3\times 3.5$m$^2$ rectangular space. We limited the turning speed of the robots between $[-0.3, 0.3]$ rad/s and the linear speed between $[0,0.1]$ m/s. The control parameter $ \beta$ is set to $0.05$. All the robots 
know their own initial positions under the same $/world$ frame. The target of the robots are set to $r_1=[1;0.1],r_2=[0.5;1];r_3=[1;0.2];r_4=[0.3;0.5] $.

The sub-figures (a) to (j) are the exemplar moments during the consensus process, and a recorded video is available
\url{https://youtu.be/a0k4KicX9u4}.
In each sub-figure, the left simulated graph provides the high-level reference that is generated by the proposed distributed optimisation algorithm, while the right two figures show the tracking performances of the simulated and real robots. It is noted that the robots approach the optimised reference generated by cooperative optimisation instead of their local target, and the stopping distance is set as $0.3$m to avoid the collision of robots. 

{In the real-world experiment, there are errors in the robot system dynamic information (position) of the vehicle due to the interference of external factors, such as different ground friction. Nevertheless, the proposed algorithm can deal with them and update its local information by communicating with adjacent agents to optimise the global cost.} From the figures and video, we observe that the proposed algorithm always optimises the global optimal position for each robot. The robots reach the consensus of the final target by the proposed approach. The problem under our consideration covers a wider range of cooperative control problems, including initialisation-dependent consensus problems, and collaborative target tracking and search problems. The deployment and realisation of the proposed algorithm on real robotic systems demonstrate its significant potential for real-world applications.

\section{CONCLUSIONS}\label{sec_con}
This paper proposes an output regulation-based distributed optimisation algorithm in the context of multi-agent cooperative control. It optimises the local objectives of the agents by letting them communicate with their neighbours, and in the meantime ensures that the agents can  reach the global optimal. 
The  algorithm can track the global target, instead of converging to a centre based on initial positions, and moreover, it can handle multi-agent systems with linear and heterogeneous dynamics.
Both theoretical guarantee and experimental validations are studies, which will promote the future deployment of distributed cooperative optimal control to solve real world applications with guaranteed convergence and optimality. 





 \section*{ACKNOWLEDGMENT}
\includegraphics[height=8pt]{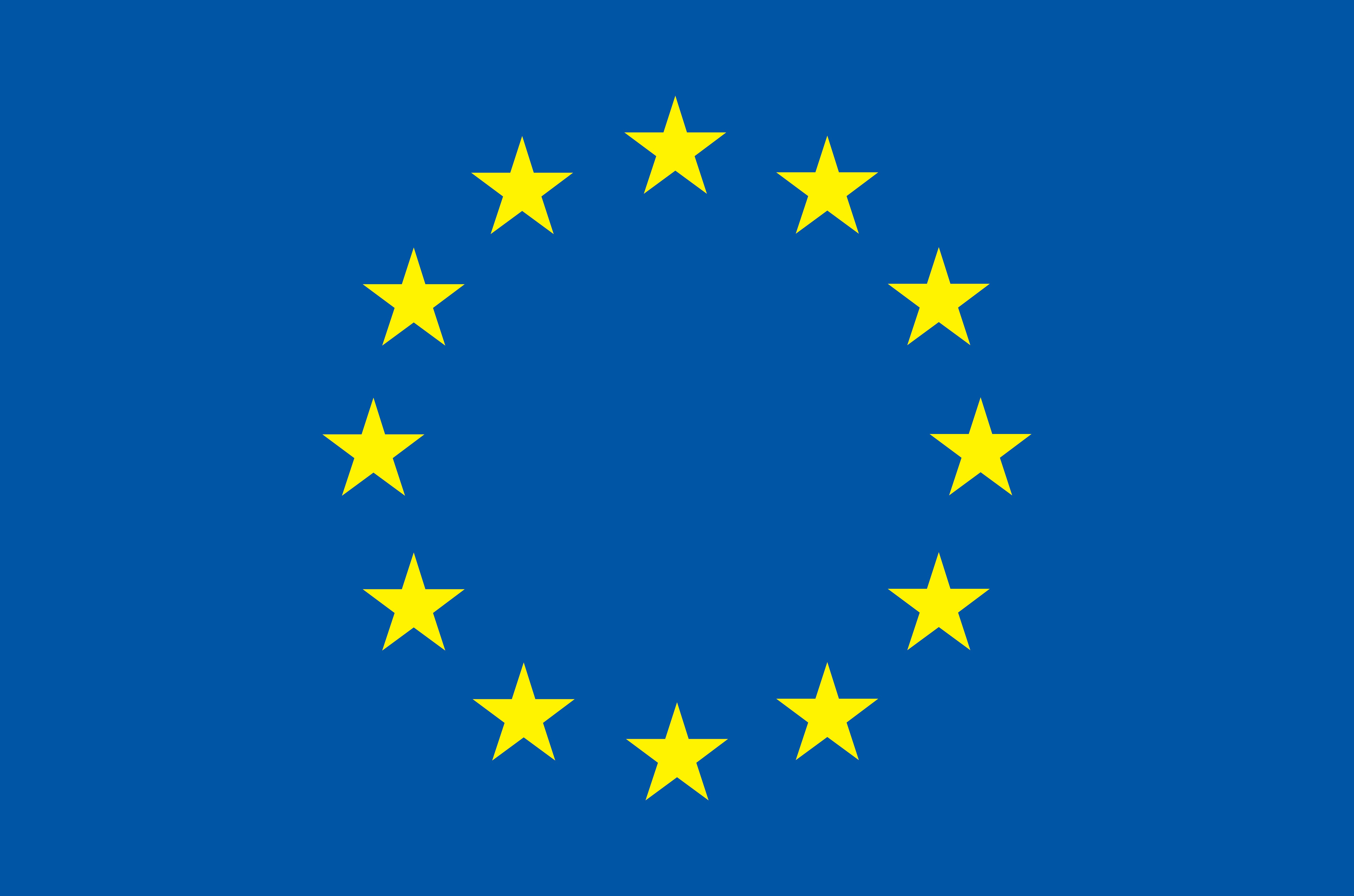} This project has received funding from the 
European Union’s Horizon 2020 research and innovation programme under grant 
agreement No 956123, and is also supported by the UK EPSRC under project [EP/T026995/1].





\bibliographystyle{ACM-Reference-Format} 
\bibliography{sample}


\end{document}